\documentclass[a4wide,11pt]{amsart}
\usepackage[margin=1.0in]{geometry} 
\usepackage{multicol}
\usepackage{epsfig}
\usepackage{latexsym}
\usepackage{epic}
\usepackage[leqno]{amsmath}
\usepackage{amssymb}
\newtheorem{theorem}{Theorem}
\newtheorem{lemma}[theorem]{Lemma}

\newtheorem{proposition}[theorem]{Proposition}
\newcommand{\PP}{{\mathbb P}}

\newcommand{\cE}{{\mathcal E}}

\title{Identifying a species tree subject to random lateral gene transfer}

\author[]{Mike Steel$^1$, Simone Linz$^{1,2}$, Daniel H.  Huson$^2$ and  Michael J. Sanderson$^3$}

\begin{document}

\begin{abstract}
A major problem for inferring species trees from gene trees is that evolutionary processes can sometimes favour gene tree topologies that conflict with an underlying species tree. In the case of incomplete lineage sorting,
this phenomenon has recently been well-studied, and some elegant solutions for species tree reconstruction have been proposed. One particularly simple and statistically consistent estimator of the species tree  under incomplete lineage sorting is to combine three-taxon analyses, which are phylogenetically robust to incomplete lineage sorting. In this paper, we consider whether such an approach will also work under lateral gene transfer (LGT).  By providing an exact analysis of some cases of this model, we show that there is a zone of inconsistency for triplet-based species tree reconstruction under LGT. However, a triplet-based approach will consistently reconstruct a species tree under models of LGT, provided that the expected number of LGT transfers is not too high.  Our  analysis involves a novel connection between the LGT problem and random walks on cyclic graphs.  We have implemented a procedure for reconstructing trees subject to LGT or lineage sorting in settings where taxon coverage may be patchy and illustrate its use on two sample data sets. \end{abstract}

\maketitle

{\em Address:} $^{1}$Allan Wilson Centre for Molecular Ecology and Evolution,
University of Canterbury, Christchurch, New Zealand

$^2$ZBIT, Universit{\"a}t T{\"u}bingen  C310a, Sand 14, 72076 T{\"u}bingen, Germany.

$^3$Department of Ecology and Evolutionary Biology, University of Arizona, 1041 E. Lowell, Tucson, AZ, USA.


\newpage
\section{Introduction}

Phylogenetic trees inferred from different genes often suggest different evolutionary histories for the species from which they have been sampled. This problem of gene tree `incongruence' is widely recognized in molecular systematics \cite{gal2}.  It is  particularly relevant to the question of the extent to which the history of life on earth can be represented by a 
phylogenetic tree, rather than a complex network of reticulate evolutionary events, such as  species hybridization, lateral gene transfer (LGT) and endosymbiosis.  There are several well-recognized causes of gene tree incongruence, most of which apply even in the absence of reticulate evolution, and we begin by discussing these.

Firstly, there is always an  expected amount of disagreement under any model of tree-based Markovian evolution, simply due to random sampling effects (i.e. the sequences are of finite rather than infinite length); moreover, this effect becomes magnified as branches in the tree become very short, or very long \cite{mar}.  Furthermore, regardless of how much data one has,  certain tree reconstruction methods may exhibit systematic errors, due to phenomena such as long branch attraction, or where the model assumed in the analysis differs significantly from the process that generated the data (`model mis-specification') \cite{fels}.
 
A second basic reason for gene trees to differ in topology from the underlying species tree is the population-genetic phenomenon of  incomplete lineage sorting.  As one traces the history of a gene sampled
from different extant species back in time, the resulting tree of coalescent events can differ from the species tree within which these lineages lie.  Recent theoretical work \cite{ros} based on the multi-species coalescent  has shown that the most probable gene tree topology can differ from the species tree topology,  when the number of taxa is greater than three. By contrast, it has long been known that for triplets, the matching topology is the most probable topology \cite{nei}, \cite{taj}.   Gene tree discordance due to incomplete lineage sorting is well-established in many data sets, and has been proposed as an explanation for why, for example, up to 30\% of the gene trees in the tree 
((human, chimp), gorilla) do not support this species relationship \cite{hob}.   Further recent work on lineage sorting has investigated the statistical consistency of building species trees from gene trees (or the clades they contain) according to various consensus criteria \cite{all}, \cite{deg}.   Additional reasons for gene tree discordance that are still consistent with a species tree are gene duplication and loss, and recombination.

If we turn now to reticulate evolution, it helps to distinguish between two types: hybridization (which will include, for example, endosymbiosis,  the transfer of a sizable 
percentage of the genome of one species into another  or the combination of two genomes into a larger genome)  and LGT (which is widespread in bacteria and includes the transfer of one or a small number of genes from one organism to another).  In the case of hybridization, it is clear that no single tree can adequately describe the evolution of the taxa under study, and that a network (or a set of species trees) is  usually a more appropriate representation. 

Hybridization will lead to gene tree incongruence but it leaves a statistically different signature to the processes of lineage sorting or sampling error discussed above.   For example, in the case of lineage sorting,  for  a given triplet of taxa, we expect  that one of the two topologies will be well supported,
and the other two topologies will have lower but approximately equal support.  On the other hand, under hybridization, we expect to find support for two of the three topologies that reflect the hybridization event but little support for the third).  A number of authors have explored this question of distinguishing hybridization from lineage sorting  \cite{chu}, 
 \cite{hol}, \cite{hold}, \cite{jol}, \cite{yu}.

The second type of reticulate evolution, LGT, is the main concern of this paper, and is particularly relevant for prokaryotic evolution \cite{dag}, \cite{doo}, \cite{jai}.   A fundamental and much-debated question is whether a species tree can be reconstructed from gene trees if genes are randomly transferred between the lineages of the tree \cite{abb}, \cite{sz}.  One viewpoint holds that if the vast majority of genes have been transferred during their history then few gene tree topologies will agree with any species tree topology, so  it makes little sense to talk about a single species tree \cite{bap}, \cite{dag2} (however, the same claim could be made, in error, for lineage sorting, as argued by  \cite{gal2}).  An alternative view is that one can still recover statistical support for a central species tree even in the presence of relatively high rates of LGT (\cite{abb}, \cite{roc}, \cite{sz}), particularly if these transfers occur in a mostly random (rather than concerted) fashion.

Random models for LGT have been proposed and studied by a number of authors,  particularly  \cite{gal1}, \cite{lin}, \cite{roc} and  \cite{suc}.   These models are somewhat similar -- random LGT events occur according to a Poisson process, and the main differences concern whether the rate of transfer between two points in the tree is constant or dependent on the phylogenetic distance between them.  In  \cite{roc}, the most recent of these papers, Roch and Snir establish a strong transition result, which shows how the species tree can be reconstructed from 
a given (logarithmic) number of gene trees, provided that the expected number of LGT events  lies below a certain threshold;  above this threshold, it becomes impossible to distinguish the underlying species tree from alternative trees based only on the given gene trees.  Our results are complementary to this work, as our interest is more in the statistical consistency of species tree reconstruction  and, in particular the consistency of tree reconstruction of triplets in a larger tree.

In our paper we begin by setting up some definitions to formally describe the way in which an arbitrary sequence of LGTs on a species tree determines the topology of the associated gene tree.   We consider the combinatorial aspects of this process for any given triplet of leaves. 
We then introduce the model of random LGT events from \cite{lin} along with an extension to allow the rates of LGT to vary with time and with phylogenetic distance. 
Under these models, we provide an exact analysis of this model on three-taxon and four-taxon trees, showing that for three-taxon trees, the species topology always has strictly higher probability than the other two competing topologies, but for certain four-taxon trees there is a zone of (weak) statistical inconsistency. In Section~\ref{general}, we consider trees with an arbitrary number of leaves and establish a sufficient condition for statistically consistent species tree reconstruction  from gene trees.   Essentially, this condition is an upper bound on the expected number of transfers of certain types in the tree.  

We then discuss how estimates of the rate of LGT in a large tree could impact on this analysis, and consider the difficult problem of reconstructing a tree when the topology of each gene tree can be influenced by both LGT and incomplete lineage sorting processes.  Finally, we describe and illustrate a simple algorithm for reconstructing a species tree from gene trees which may have patchy taxon coverage; here tree reconstruction is statistically consistent if each gene evolves under the random model of LGT or incomplete lineage sorting, provided the rate of LGT is sufficiently low.  We end with a brief discussion and some questions for further work.

\section{Combinatorial LGT analysis}

\subsection{Definitions}
Throughout this paper $X$ will denote a set of species of size $n$, and $A$ will denote a subset of $X$ of size $3$. 
Consider  a rooted phylogenetic  `species' tree $T$,  with leaf set $X$, and a vertex $\rho$.  
We will regard $T$ as a 1-dimensional simplicial complex  (i.e. the edges as intervals) so each `point' $p$ in $T$ is either a vertex or an element of the interval that corresponds to an edge. 
Consider a coalescence time scale: 
$t: T \rightarrow [0, \infty)$ of the tree with the coalescence time increasing into the past.  Then: 
\begin{itemize}
\item $t(p) = 0 \Leftrightarrow p$ is a leaf, 
\item If $u$ is a descendant of $v$ then $t(u) < t(v)$.
\end{itemize}
We refer to $t(p)$ as the {\em $t$-value of $p$}  and to $t(\rho)$ as the {\em timespan} of the tree (the time from the present to the most recent common ancestor (MRCA) of all the species in $X$).

A {\em lateral gene transfer (LGT) on $T$} (or, more briefly a {\em transfer event}) is an arc from $p \in T$ to $p' \in T$ where $p$ and $p'$ are contemporaneous, i.e.  $t(p)=t(p')$. We will also 
assume that neither $p$ nor $p'$ are vertices of $T$ (i.e. transfers go between points on the edges of the tree).

We write $\sigma = (p, p')$ to denote this transfer event and we write $t(\sigma)$ for the common value of
$t(p)$ and $t(p')$. We will assume that no two transfer events occur at \underline{exactly} the same time. 

Let $\underline{\sigma}= \sigma_1, \ldots, \sigma_k$ be a sequence of transfer events ($\sigma_i = (p_i, p'_i)$) arranged in increasing $t$-value order  (thus $\sigma_1$ is the most recent transfer event and $\sigma_k$ is the most ancient; moreover, the total ordering is well defined by the assumption that no two transfer events took place at the same date). Thus, we will always assume  in what follows that:
$$0<t(\sigma_1) < t(\sigma_2) < \cdots < t(\sigma_k)< t(\rho).$$

We refer to $\underline{\sigma}$ as a {\em transfer sequence} on the species tree $T$.   
In the biological context, we view $\underline{\sigma}$ as describing the transfer history of a particular gene,  so different genes will have different associated 
transfer sequences (including, possibly, the empty transfer sequence, if no transfer events occur on $T$). 

Given a species tree $T$ with the leaf set $X$  and a transfer sequence $\underline{\sigma}  = \sigma_1, \ldots, \sigma_k$ on $T$, we obtain an associated {\em gene tree} $T[\underline{\sigma}]$.  To describe this tree more precisely, 
we assume, as in \cite{lin}, that an LGT arc from point $p$ to $p'$ replaces the gene that was present on the edge at $p'$ with with the transferred gene from $p$.  Thus, if we trace the history of the gene from the present to the past (i.e. in increasing coalescence time), each time we encounter an incoming horizontal arc into this edge we follow this arc (against the direction of the arc).  In this way, the species
tree, along with any sequence of LGTs describes an associated gene tree.   We can formalize this mathematically as follows. 
 For a transfer sequence $\underline{\sigma} = \sigma_1, \ldots, \sigma_k$  where $\sigma_i =(p_i, p'_i)$, consider the tree $T$ together with a directed edge for each $\sigma_i$ placed between $p_i$ and $p'_i$ for each $i \in \{1, \ldots, k\}$
and regard this network as a one-dimensional simplicial complex.  Now for each $i \in \{1, \ldots, k\}$, delete the interval of this 1-complex  immediately above $p'_i$ and consider the minimal connected subgraph of the resulting complex that contains $X$.  Call this tree $T[\underline{\sigma}]$.   An example is shown in Fig.~\ref{fig1}.
\bigskip

\begin{figure*}[ht]
\center
\resizebox{12cm}{!}{
\includegraphics{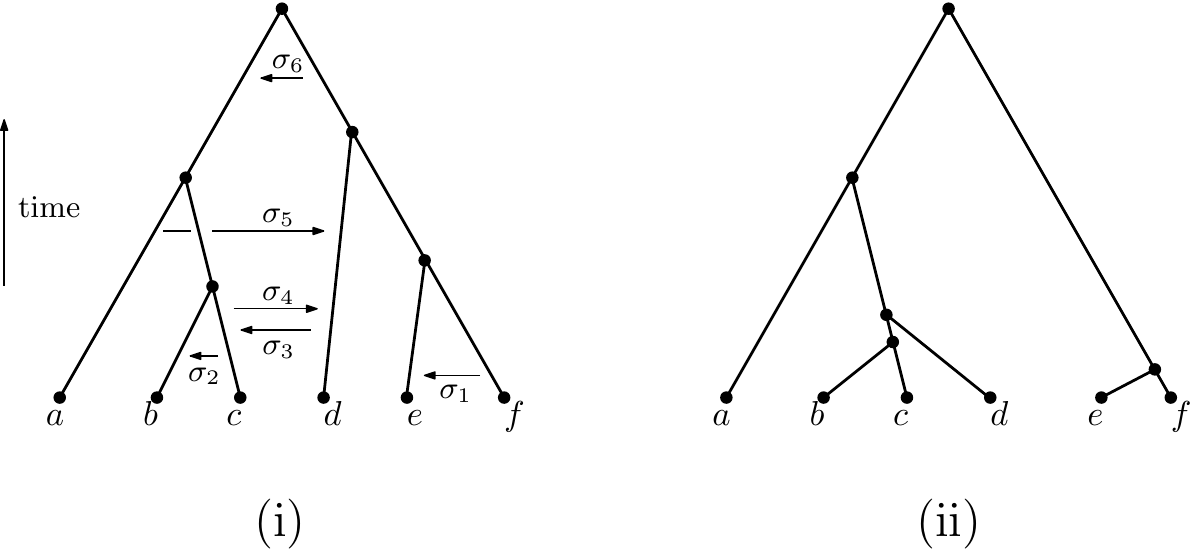}
}
\caption{(i) A rooted binary $X-$tree $T$ with a sequence $\underline{\sigma}$ of six transfer events, labelled in increasing order into the past; (ii) The tree $T[\underline{\sigma}]$.  In this example, $\underline{\sigma}$  induces a match for  $a,b,c$ and a mismatch for $a,b,d$.}
\label{fig1}
\end{figure*}

\begin{itemize}
\item Given the pair $T, \underline{\sigma}=\sigma_1, \ldots, \sigma_k$,  define the following sequence of derived $X-$trees, for $0 \leq r \leq k$:
$$T_0=T, \mbox{  } \mbox{  } T_r = T_{r-1}[\sigma_r].$$
Thus, $T_k = T[\underline{\sigma}]$, and for $1 \leq r <k$, $T_r = T[\sigma_1, \ldots, \sigma_r]$ is the gene tree obtained from $T$ by performing just the $r$ most recent transfers.

\item Given $T' \in \{T_0, T_1, \ldots, T_k\}$,  a point $p \in T'$  and any non-empty subset $Y$ of $X$, let ${\rm des}_Y(T',p)$ denote the subset of $Y$ whose elements are the {\em descendants} of $p$  (i.e which become separated from the root of $T'$  if $p$ is deleted). 

 \end{itemize}

\section{Triplet analysis}

Consider  a species tree $T$ on $X$. For the rest of this paper, we will assume that $T$ is binary (i.e. fully resolved). Let $A = \{a,b,c\}$ be a subset of $X$ of size 3. We write $T|A$ to denote the 
binary phylogenetic tree on leaf set $A$ that is induced by restricting the leaf set of $T$ to $A$ \cite{sem}.  We refer to $T|\{a,b,c\}$ as a {\em triplet (species tree) topology},
and we write $a|bc$ to denote the triplet topology in which the root of this tree separates leaf $a$ from the pair $b,c$.

\subsection{Key triplet definitions:} 

For a sequence $\underline{\sigma}$ of transfer events, we say that:
\begin{itemize}
\item
$\underline{\sigma}$ induces a {\em match} for $A=\{a,b,c\} \subseteq X$ if the species tree $T$ and its associated gene tree $T[\underline{\sigma}]$  resolve $a,b,c$ as the same three-taxon tree; i.e., if 
$$T[\underline{\sigma}]|A = T|A.$$
Otherwise,  if $T[\underline{\sigma}]|A$  is one of the other two rooted binary tree topologies on leaf set $A$, we say that $\sigma$ induces a {\em mismatch} for $A$.
An example is provided in Fig.~\ref{fig2}.
\bigskip

\item
For a transfer event $\sigma = (p,p')$, we say that $\sigma$ is {\em into an $A-$lineage} if ${\rm des}_A(T, p')$ is a single element of $A$ (two such transfers are shown in Fig.~\ref{fig2}).

 \end{itemize}

Now, suppose that $\underline{\sigma} =\sigma_1, \ldots, \sigma_k$ is a sequence of transfers on $T$.  
Consider the resulting sequence $(T_r; 0 \leq r \leq k)$ of derived trees and let   $\sigma_r= (p_r, p'_r)$. 

\begin{itemize} 

\item 
If ${\rm des}_A( T_{r-1}, p_r')  = \{x\},$
for some $x \in A$, we say that $\sigma_r$ is an {\em $A-$transfer}, and that it {\em transfers $x$}.

\item
If $\sigma_r$ transfers $x$  and  ${\rm des}_A( T_{r-1}, p_r)  = \emptyset$, we say that $\sigma_r$  {\em moves} $x$ and we refer to $\sigma_r$ as an {\em $A-$moving transfer}.

\item If $\sigma_r$ transfers $x$, and if ${\rm des}_A( T_{r-1}, p_r)  = \{y\}$, we say that $\sigma_r$ {\em joins $x$ to $y$} and we 
 refer  to $\sigma_r$ as an {\em $A-$joining transfer.}

\end{itemize}

Note that any $A-$transfer is either a moving or joining transfer (but not both).  Examples of $A-$transfers are shown in Fig.~\ref{fig2}.
Notice also that the first $A-$transfer is always a transfer into an $A-$lineage, but later ones need not be. Moreover, a transfer into an $A-$lineage
may not be an $A-$transfer if certain other $A-$transfers proceed it in some transfer sequence.

\begin{figure*}[ht]
\center
\resizebox{8cm}{!}{
\includegraphics{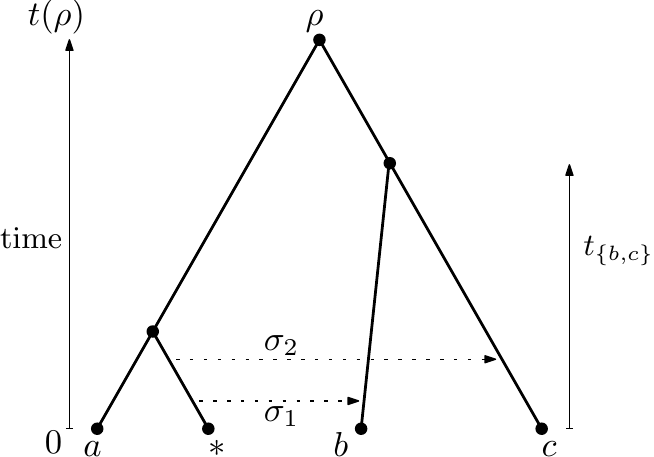}
}
\caption{A tree subject to a sequence $\underline{\sigma} = \sigma_1, \sigma_2$ of two transfer events, which induces a match for $A$, even though each of its component transfers by itself would induce a mismatch topology.  The transfer event $\sigma_1$ moves $b$ (and so is an $A-$moving transfer) while $\sigma_2$ joins $c$ to $b$ (and so is an $A-$joining transfer). Note that if $\sigma_1$ was removed then $\sigma_2$ would become an $A-$moving transfer.}
\label{fig2}
\end{figure*}

\subsection{Triplet combinatorics under LGT transfers}

We now state two combinatorial lemmas which will form the basis for the stochastic analysis that follows later.

Let $t_A$ denote the  time from the present to the MRCA in $T$ of  the closest pair of taxa in $A$ (so if $T|A= a|bc$ then $t_A$ is the time from
the present to the MRCA of $b$ and $c$).

The following lemma collects for later reference some observations concerning the impact of different types of transfers.

%%%%%%%%%%%%%%%%%%%%%%%%%%%%%%%%%%TWO_LEMMAS %%%%%%%%%%%%%%%%%%%%%%%%%%

\begin{lemma}
\label{lem1}

Let $\underline{\sigma} = \sigma_1, \ldots, \sigma_k$ be a sequence of transfer events on a rooted binary $X-$tree $T$ and let $A=\{a,b,c\} \subseteq X$.

\begin{itemize}
\item[(a)] 
If $\underline{\sigma}$ induces a mismatch for $A$, then $\underline{\sigma}$ must contain an $A-$transfer with a  $t-$value less than $t_A$.
\item[(b)] Moreover,  precisely one of the following occurs:

\begin{itemize}

\item[(i)]
$\underline{\sigma}$ has no $A-$transfers.  In 
this case, $\underline{\sigma}$ induces a match for $A$. 

\item[(ii)]
$\underline{\sigma}$ contains at least one $A-$joining transfer.  In this case, if  the first such  transfer in $\underline{\sigma}$  joins $x$ to $y$ then
$T[\underline{\sigma}]|A = z|xy$ where $\{x,y,z\} = A$.

\item[(iii)] 
$\underline{\sigma}$ has no $A-$joining transfers, but  it has an $A-$moving  transfer with a $t-$value less  than $t_A$.   In this case, if $\sigma_r$ denotes the first
such $A-$moving transfer in $\underline{\sigma}$ then:  $$T[\underline{\sigma}]|A=T[\sigma_r,\ldots, \sigma_k]|A.$$
\end{itemize}
\end{itemize}
\end{lemma}

\bigskip
A second combinatorial lemma, extends case (b)(iii)  of Lemma~\ref{lem1} slightly and is used in the proof of Theorem~\ref{fourcasethm}. In order to state it, we  need some further definitions. 

Suppose $\underline{\sigma} = \sigma_1, \ldots, \sigma_k$ is a sequence of transfer events on a rooted binary $X-$tree $T$ with $t(\sigma_k) < t_A$, and with no $A-$joining transfers.   Construct an associated sequence of trees $T'_0, T'_2, \ldots, T'_k$  as follows.
Set $T_0'=T$ and construct by $T'_{i+1}$ from $T'_i$ by the following procedure.  If $\sigma_i$ is not $A-$moving then set $T'_{i+1}=T'_i$. 
 If $\sigma_i = (p_i, p'_i)$  moves $x \in A = \{a,b,c\}$
then let $T'_{i+1}$ be the tree obtained from $T'_i$ by:
\begin{itemize}
\item[(i)] deleting all $p \in T'_i$ with $t(p) < t(\sigma_i)$, 
\item[(ii)] labeling $p_i$ by $x$,
\item[(iii)] For each $z \in A-\{x\}$, assigning label $z$ to  the unique point $p_z$ of $T'_i$ that has $t(p_z) = t(\sigma_i)$ and $z \in {\rm des}_A(T'_i, p_z)$.
\item[(iv)] We will regard the other leaves in the tree as unlabeled.
\end{itemize}
Finally transfer the time dating from $T_i$ across to $T_{i+1}$.

\begin{lemma}
\label{lem2}
Suppose $\underline{\sigma} = \sigma_1, \ldots, \sigma_k$ is a sequence of transfer events on a rooted binary $X-$tree $T$ with $t(\sigma_k) < t_A$, and with no $A-$joining transfers. Then $T[\underline{\sigma}] |A = T'_k|A.$
\end{lemma}

\bigskip

\section{Statistical signals for the central tree under LGT and incomplete lineage sorting}

We begin by recalling the model of LGT described in \cite{lin}, which made the following assumptions:  (1) a binary, labeled, rooted and clocklike species tree $T$ is given, as well as all the splitting
times along this tree;  (2) differences between a gene tree and $T$ are only caused by LGT events; (3) the transfer rate is homogeneous per gene and unit time; (4) genes are transferred independently; (5) one copy of the transferred gene still remains in the donor genome; (6) the transferred gene replaces any existing orthologous counterpart in the acceptor genome.  
We refer to this model from \cite{lin} as the {\em standard LGT model} and we will consider extensions of it which relax assumptions (2) and (3).  In particular,  consider the following relaxation of assumption (3)  in which transfer events on $T$ occur as a Poisson process through time, in which the rate of transfer event from  point  $p$ on a lineage to a contemporaneous point $p'$ on another lineage at time $t$  occurs at the rate $f(d(p, p'), t)$ where $f(d, t)$ is a constant or at least monotone non-increasing function in $d$ (though it can vary non-monotonically in $t$) and $d(p, p')$ is the evolutionary distance in the tree between contemporaneous points $p$ and $p'$ in the tree.
We call this model the {\em extended LGT model}.

Both the standard and extended LGT models induce a well-defined probability distribution on gene tree topologies, both for the original set of taxa, and for any subset. 

In the standard LGT model, the number of transfers has a Poisson distribution, with a mean equal to the rate of LGT transfer out of any given point in the tree times
the sum of the branch lengths (phylogenetic diversity) of the tree.    We first establish that the Poisson distribution for the number of LGT transfers (in total or just the
transfers in an $A-$lineage) still holds for the extended LGT model.

\begin{lemma}
\label{lempois}
Under the extended LGT model, the total number of transfers, and the number of transfers into an $A-$lineage (for any given subset $A$ of $X$ of size 3) each have a Poisson distribution.
\end{lemma}

\begin{proof}

We begin by recalling a general property for any collection $(P_1, \ldots, P_k)$ of independent Poisson processes, where process $P_i$ has intensity $r_i(t)$. Let  $Y_i(t_0)$ count the number of times process $P_i$ occurs up to time $t_0$, and let $Y= \sum_{i=1}^k Y_i(t_0)$. Then $Y_i(t_0)$ has a Poisson distribution with mean $\int_0^{t_0} r_i(t)dt$. Moreover,  since the sum of independent Poisson random variables has a Poisson distribution, with a mean equal to the sum of the individual means, it follows that $Y$ has a Poisson distribution, with mean $\sum_{i=1}^k \int_0^{t_0} r_i(t)dt.$ (for background on these stochastic  results, the reader may wish to consult \cite{gri}).  We apply this result as follows: let $N$ be the total number of transfers and $N_A$ the number of transfers into an $A-$lineage in the tree.  
We can express $N$ by considering  all 
intervals $I$ between speciation events in $T$ and, within each interval, consider all ordered pairs of lineages (restricted to that interval) $l_1, l_2$. For each such ordered triplet $i = (I, l_1, l_2)$, let $P_i$  be the  Poisson process of transfers from $l_1$ to $l_2$ (which may depend non-homogeneously on time) and 
 $N_i$ denotes the number of these transfers for this triplet.  Then the $P_i$ are independent processes, and  $N = \sum_{i=(I, l_1, l_2)} N_i$, so $N$ has a Poisson distribution.  
 
Regarding $N_A$, we consider all intervals $I$ between speciation events in $T$ and, within each interval, consider all ordered pairs of lineages (restricted to that interval) $l_1, l_2$, where $l_2$ has exactly one of $a,b,c$ as a descendant. Let  $P_i$  be the  Poisson process of transfers from $l_1$ to $l_2$ (which may depend non-homogeneously on time) and let  $N'_i$ denote the number of such transfers for this triplet.  Note that the $P_i$ are independent processes, and that
$N_A=\sum_{i=(I, l_1, l_2)} N'_i$, and so again has a Poisson distribution.  

\end{proof}

Although the process of LGT transfers (under the standard or extended) model is a continuous time process, there is an associated discrete process that induces an identical distribution on gene trees; it is obtained by considering the decomposition described in the proof of Lemma~\ref{lempois}.  Thus with any sequence $\underline{\sigma} = \sigma_1, \ldots, \sigma_k$  of LGT transfers on $T$, we may associate the  discrete sequence $\underline{s} = s_1, \ldots, s_k$ where $s_i$ refers to a triple $(I, l_1, l_2)$, in which case
$T[\underline{\sigma}] = T[\underline{s}]$ (i.e. all that matters in determining the resulting tree topology is the sequence of transfers between branches of the tree
and their relative ordering, not the actual times that they occur).  Consequently, each such discrete sequence $\underline{s}= s_1, \ldots, s_k$ for $T$ has a positive probability, and under the standard LGT model, this probability has a Poisson distribution that just depends on $k$.

%%%%%%%%%%%%%%%%%%%%%%%%%%%

\section{Exact analysis of triplets in  three- and four-taxon trees}

For a transfer sequence $\underline{\sigma}$ generated by this type of LGT process, and any triplet $a,b,c \in X$, we are interested in the probability that  $\underline{\sigma}$ induces a match for $a, b,c$.

For three leaves, an exact  analysis is  straightforward, as we now show. 

\begin{proposition}
\label{lems}
If  $T$ has just three taxa, then under the extended LGT model, the probability that a transfer sequence  induces a match for the three taxa is strictly greater than the probability it induces  either one of the two mismatch topologies (which have equal probability).
\end{proposition}

\begin{proof}
 When $A=\{a,b,c\}$ then there are no $A-$moving transfers, and a transfer $\sigma$ is $A-$joining if and only if $t(\sigma)< t_A$;
let $N$ denote the number of such transfers.
We can express $N$ as the sum of six 
random variables, which counts the number of transfers from the lineage $x$ to lineage $y$ (for $x,y \in A, x\neq y$).
 By the assumptions of the model, and  Lemma~\ref{lempois},  $N$ has a Poisson distribution
with some fixed mean $m$. Now if $N=0$ we obtain a matching topology for the three taxa, while if $N>0$ then, by Lemma~\ref{lem1}, the topology of the tree is determined by the first
transfer (since it is, by necessity, $A-$joining) and under the assumption of the model (in particular invoking the property that $f$ is non-increasing), the probability that this
first transfer is between two of the most closely related taxa is at least $\frac{1}{3}$.  Thus, 
 the probability of a matching topology for the three taxa  is at least:  $$\PP(N=0) + \frac{1}{3}\PP(N>0) = e^{-m}+ \frac{1}{3} (1-e^{-m}),$$
 while for either mismatched topology, the probability is equal to the other mismatch topology and is no more than
$\frac{1}{3} (1-e^{-m}).$ This completes the proof.
\end{proof}

\bigskip

\subsection{Four-taxon case}

The analysis of the distribution of triplet gene tree topologies generated by random LGT transfers on a four-taxon tree is considerably more interesting than the three-taxon case.  Note that with four taxa, there are two rooted binary tree shapes -- the `fork-shaped'  tree (with two cherries,  as shown in Fig.~\ref{fig4}) and the pectinate tree (with one cherry). We will provide an exact analysis for the first of these tree shapes under the standard LGT model, as this suffices to demonstrate a zone of statistical inconsistency, though we discuss briefly how an analogous (but more complex) analysis could be carried out for the other rooted tree shape.  Our analysis requires no approximations, nor any imposition of an upper bound on the number or rate of LGT transfers.  It relies on associating a  random walk on a six-cycle to the LGT process.

We begin with some definitions.    For four  leaves $x,y,z,w$, we write $(x,y; w,z)$ to denote the rooted binary tree, with the leaves $x,y$ on one side of the root and $w,z$ on the other,
and with the MRCA of $x,y$ at a fixed  time $t_{\{x,y\}}$ and the MRCA $w,z$ at a later fixed time $t_{\{w,z\}}$. Thus $(x,y;w,z) = (y,x;w,z) = (y,x;z,w) = (x,y;z,w)$, but no other symmetries hold.
For example, the tree $(a,*; b,c)$  is shown in Fig.~\ref{fig4}(i) and the tree $(b,c; a,*)$ in Fig.~\ref{fig4}(ii). Here $*$ refers to the fourth taxon, the identity of which plays no role when we come to consider the topology of the triple $a,b,c$.

\bigskip

We now state the main result of this section.

%%%%%%%%%%%%%%%%%%%%%%%%%%%%%%%%%%%%%%%%%%%%%%%%
\begin{theorem}
\label{fourcasethm}

Suppose $T$ is a rooted four-taxon tree and $A = \{a,b,c\}$ is a subset of the leaf set $X$ of $T$, and suppose that $T|A = a|bc$ (thus $T|A$ is a tree of type $\tau'_a$ or $\tau_a$ in Fig.~\ref{fig4}). Let  $\PP(x|yz) = \PP(T|A=x|yz)$ under the standard LGT model of \cite{lin}, where $\{x,y,z\} = A$.
\begin{itemize}

\item[(i)]
Suppose  $T$  is of type $\tau_a$. Then for  any $t-$value for the MRCA of $(a, *)$, there is a sufficiently large $t$-value for the MRCA of $(b,c)$, for which  the matching gene tree topology $a|bc$ has a lower probability than either of the alternative mismatch topologies.  

More precisely, for  $\mu = \frac{1}{3}\lambda t_{\{a,*\}}$ and $B = 3\lambda (t_{\{b,c\}}-t_{\{a,*\}})$,  we have:
\begin{equation}
\label{random1}
\PP(a|bc) = \frac{1}{3}\left[1-e^{-7\mu}(1-e^{-2\mu} - e^{-B}(1+e^{-2\mu}))\right]
\end{equation}
and $\PP(b|ac)  = \PP(c|ab)$; moreover, $\PP(a|bc) < \PP(b|ac)  = \PP(c|ab)$   if and only if
 $t_{\{b,c\}}-t_{\{a,*\}}$ is  greater than $\frac{1}{3\lambda}\ln(\frac{1+e^{-2\mu}}{1-e^{-2\mu}}).$

\item[(ii)]
Suppose  $T$  is of type $\tau'_a$. Then for  any $t-$values for the MRCAs of $(a, *)$ and $(b,c)$, the matching gene tree topology $a|bc$ has a higher probability than either of the alternative mismatch topologies.  

More precisely, for   $\tau'_a$ and $\mu = \frac{1}{3}\lambda t_{\{b,c\}} $, and $B = 3(\lambda t_{\{a, *\}}-t_{\{b,c\}})$, we have:
\begin{equation}
\label{random2}
\PP(a|bc) =  \frac{1}{3}\left[1+e^{-7\mu}(1+e^{-2\mu} - e^{-B}(1-e^{-2\mu}))\right].
\end{equation}
In this case, $\PP(a|bc)> \PP(b|ac) (= \PP(c|ab))$ for all values of $t_{\{a,*\}}$ and  $t_{\{b,c\}}-t_{\{a,*\}}$.
\end{itemize}

\end{theorem}
%%%%%%%%%%%%%%%%%%%%%%%%%%%%%%%%%%%%%%%%%%%%%%%%%

\begin{figure*}[ht]
\center
\resizebox{10cm}{!}{
\includegraphics{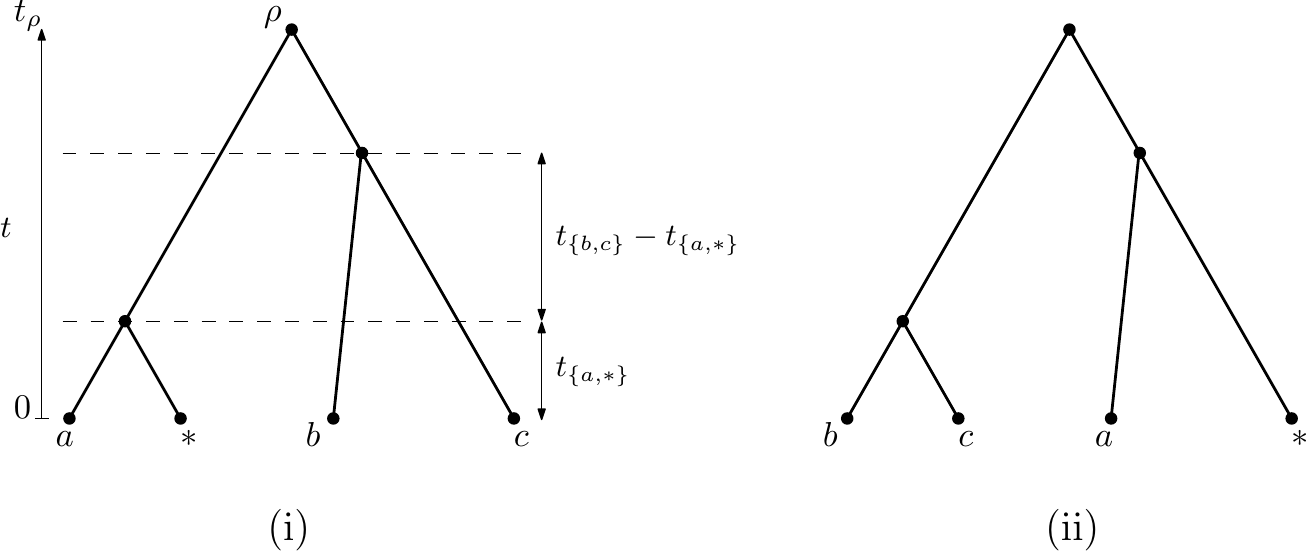}
}
\caption{(i) The species tree topology $\tau_a = (a,*; b,c)$ for which the standard random LGT model confers, for any value of $t_{\{a, *\}}$ and a sufficiently large value of 
$t_{\{b,c\}}$,  a higher probability for each of the the gene tree mismatch topologies $b|ac$ and $c|ab$ than for the matching species tree 
topology $a|bc$.  (ii) The species tree topology $\tau'_a = (b,c;a,*)$ which, in contrast to (i),  always  has higher probability under the standard LGT model for a matching gene topology $a|bc$ 
than for either mismatch topology.}
\label{fig4}
\end{figure*}

\bigskip
\begin{proof}  

First observe that,  by symmetry, we have:
\begin{equation}
\label{sym}
\PP(b|ac) = \PP(c|ab),
\end{equation}
under the standard LGT model (indeed this holds here even under the extended LGT model).

Consider the cyclic graph whose nodes are the six trees: 
$$\tau_a = (a,*;b,c), \mbox{ }   \tau'_a = (b,c;a, *),\mbox{ }   \tau_b = (b,*; a,c),  \mbox{ }  \tau'_b = (a,c; b,*),  $$ 
$$   \tau_c = (c,*;a,b), \mbox{ }  \tau'_c= (a,c;c,*),$$
which are connected into a cycle as shown in Fig. \ref{fig:RandomWalk}.

\begin{figure*}[ht]
\center
\resizebox{10cm}{!}{
\includegraphics{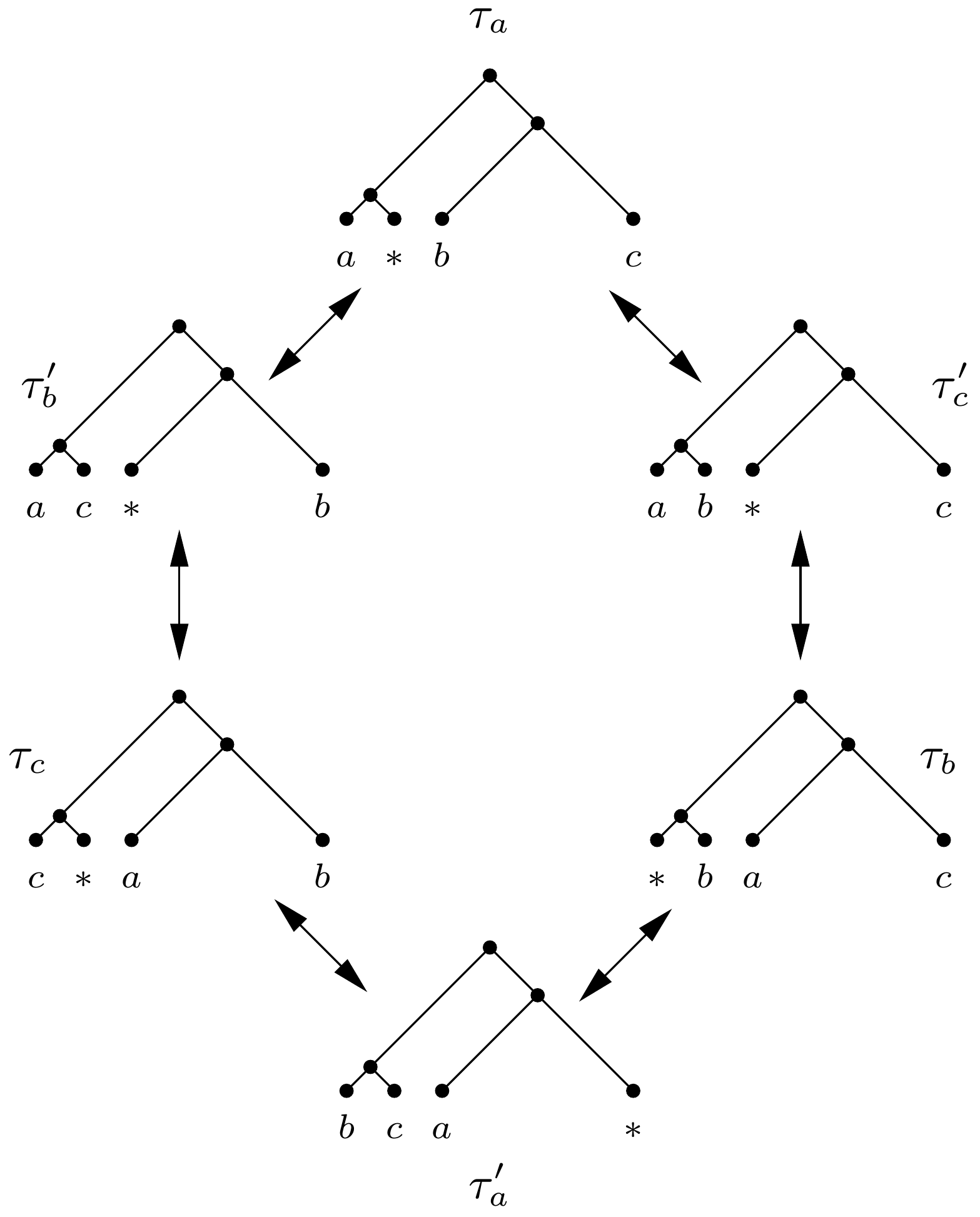}
}
\caption{Conditional on there being no $A-$joining transfers between the present and $t_A$, the topology of the tree sequence $T'_i$  induced by $A-$moving transfers in this interval is equivalent to a simple random walk on the six-cycle graph shown. }
\label{fig:RandomWalk}
\end{figure*}

Now, let $Z_t: t\geq 0$ be a continuous-time symmetric random walk on this 6-cycle graph, where the instantaneous rate of moving from one node  to either given
neighboring node is $1$. Let $p_r(t)$ $r=0,1,2,3$ be the probability that, after running the process for time $t$, this Markov process is at a node that is graph distance $r$ (by the shortest path) from its initial state.

Note that ${\bf p} ={\bf p}(t)=  [p_0(t), p_1(t), p_2(t), p_3(t)]^t$ satisfies the system of first-order linear differential equations:
$$\frac{d}{dt}{\bf p} = B {\bf p},$$
where $B$ is the $4 \times 4$ tridiagonal matrix:

\[ \left[ \begin{array}{cccc}
-2 & 1 & 0 & 0 \\
2 & -2 & 1 & 0 \\
0 & 1 & -2 & 2\\
0 & 0 & 1 & -2 \end{array} \right]\]

and so:
\begin{equation}
{\bf p}(t) = \exp(Bt){\bf p}(0) \mbox { where } {\bf p}(0) = [1,0,0,0]^t.
\label{solution}
\end{equation}

The eigenvalues of $B$ are 0, -1, -3, -4,  and the random walk on the 6-cycle is a reversible Markov process with uniform equilibrium frequency, and so
$$\lim_{t \rightarrow \infty} {\bf p}(t) = [1/6, 1/3, 1/3, 1/6]^t,$$ and each component 
$p_j(t)$ is of the form: $$p_j(t) = a_j+ b_je^{-t} + c_je^{-3t} + d_je^{-4t},$$
for constants $a_j, \ldots, d_j$ that are determined by the eigenvectors of $B$.

Using standard matrix diagonalization techniques from linear algebra, we obtain the following solution to Eqn. (\ref{solution}):

$${\bf p}(t)= \frac{1}{3}
\left[ \begin{array}{cccc}
\frac{1}{2} & 1 & 1 & \frac{1}{2}  \\
1 & 1 & -1 & -1 \\
1 & -1 & -1 & 1\\
\frac{1}{2} & -1 & 1 & -\frac{1}{2} \end{array} \right]
\left[ \begin{array}{c}
1 \\
e^{-t} \\
e^{-3t}\\
e^{-4t} \end{array} \right]
$$

From this, one immediately obtains the following result, which will be required later.

\begin{lemma}
\label{helps}
For all $t> 0$:

$$p_1(t)-2p_3(t) = e^{-t}-e^{-3t}>0,$$ and
$$2p_0(t)- p_2(t)= e^{-t}+e^{-3t} >0.$$

\end{lemma}

We return now to the proof of Theorem~\ref{fourcasethm}.   We will establish part (i),  and indicate how the proof of part (ii) follows by a directly analogous argument.

Let $\cE$ denote the event that the random sequence of transfer events $\underline{\sigma}$ generated by the model induces a match for $A$.  
Let $J$ denote the number of $A-$joining transfers between $t=0$ and $t=t_{\{a,*\}}$.  Then $J$ has a Poisson distribution with mean $2\lambda  t_{\{a,*\}} = 6\mu$, since 
at any moment in the interval $[0, t_{\{a,*\}}]$, there are four lineages, three of which lead to leaves in $A$ (therefore, for any $x\in A$, the rate of transfer from that $x-$lineage to any lineages that also lead to $A$ is
$\lambda \cdot (2/3)$). Thus the cumulative rate of an $A-$joining transfer  is $3 \lambda \cdot (2/3) = 2\lambda$. 
Consequently: 
\begin{equation}
\PP(J>0) = 1-e^{-6\mu}.
\label{Weq}
\end{equation}
Now, Lemma~\ref{lem1} (part (b)(ii)) implies that:
\begin{equation}
\PP(\cE|J>0) = \frac{1}{3},
\label{Weq2}
\end{equation}
and by the law of total probability, Eqns. (\ref{Weq}) and (\ref{Weq2}) give:
\begin{equation}
\label{ppeq}
\PP(\cE) =  \frac{1}{3}(1-e^{-6\mu}) + e^{-6\mu}\PP(\cE|J=0).
\end{equation}

Now the $A-$moving transfers between $t=0$ and $t=t_{\{a,*\}}$ constitute a continuous-time Poisson process for which the rate at which any given $x \in A$ is moved is $\frac{1}{3}\lambda$.  
Note that this process is independent of $J$ since the source point of an $A-$joining transfer is always on a different type of lineage (having an element of $A$ as a descendant) from an  $A-$moving transfer.
As the $A-$moving process proceeds in time (from $t=0$ to $t=t_{\{a,*\}}$),  the resulting sequence of trees $T'_k$ described in the preamble to Lemma~\ref{lem2} corresponds to a simple symmetric random walk on the nodes
of the 6-cycle shown in Fig.~\ref{fig:RandomWalk}, starting with tree $\tau_a$ at time $t=0$, and where the rate of moving from one node to any particular neighboring node is $\frac{1}{3}\lambda$.  At time $t=s$, the length of two pendant edges of the $\tau-$tree  will be $t_{\{a,*\}} - s$ while the other two pendant edges have a larger length of $t_{\{b,c\}}-s$,
so we stop the process when the length of the shorter pair of pendant edges
reaches zero (i.e. at $t=t_{\{a,*\}}$).  Note that the length of the pendant edges does not affect the transition process under the standard LGT model, so we can indeed view it as a discrete-state random walk on six states (rather than on a continuum of states).

Lemma~\ref{lem2} now ensures that if $\underline{\sigma}'$ is the sequence of $A-$moving transfers between $t=0$ to $t=t_{\{a,*\}}$ then $T[\underline{\sigma}']$ 
resolves $a,b,c$ in the same way as the tree $\tau_i$ does, where $\tau_i$ is the state of the random walk on the 6-cycle at time $t_{\{a,*\}}$.
At time $t= t_{\{a,*\}}$ the random walk on the 6-cycle is at one of the following nodes:

\begin{itemize}
\item   $\tau'_a$, in which case $T[\underline{\sigma}]|A= a|bc$ (with probability 1); this does not depend on any transfer events that may  occur
after  $t_{\{a,*\}}$;

\item $\tau_a$, in which case $T[\underline{\sigma}]|A= a|bc$  
\begin{itemize}
\item  with probability 1 if there is no transfer event between $t_{\{a,*\}}$ and $t_{\{b,c\}}$, or
\item
with probability $\frac{1}{3}$ if there is at least one transfer event between $t_{\{a,*\}}$ and $t_{\{b,c\}}$;
\end{itemize}

\item $\tau_b$ or $\tau_c$, in which case $T[\underline{\sigma}]|A= a|bc$ with probability 
$\frac{1}{3}$ if there is at least one  transfer event between $t_{\{a,*\}}$ and $t_{\{b,c\}}$;

\item $\tau'_b$ or $\tau'_c$, in which case $T[\underline{\sigma}]|A\neq a|bc$  regardless of any further transfers.
\end{itemize}
In this case analysis, the probability factor  $\frac{1}{3}$ arises from  Lemma~\ref{lem1}, part b(ii).  
Note also that the  probability that there is at least one transfer event between $t_{\{a,*\}}$ and $t_{\{b,c\}}$ has probability $1-e^{-B}$, and this event 
 is independent of the random walk on the 6-cycle.

Consequently, by independence, and by combining these cases we obtain:

$$\PP(\cE|J=0) = p_3(\mu) \cdot 1 + p_0(\mu) \cdot (e^{-B} +\frac{1}{3}(1-e^{-B}) )+ p_2(\mu)\cdot \frac{1}{3}(1-e^{-B}) + p_1(\mu)\cdot 0. $$
which, combined with Eqn. (\ref{ppeq}) and Lemma~\ref{helps},  establishes Eqn. (\ref{random1}).
The remainder of part (i) is justified by  Eqn.~(\ref{sym}) and straightforward algebra to determine when the expression in Eqn.(\ref{random1}) is lower than
the probability of a specific mismatch topology on $A$.

For the proof of part (ii), an analogous argument shows that for $T|\{a,b,c\}$  of type  $\tau'_a$ and $\mu = \frac{1}{3}\lambda t_{\{b,c\}} $ and $B = 3(\lambda t_{\{a, *\}}-t_{\{b,c\}})$, Eqn. (\ref{ppeq}) still holds, that is:
$$\PP(\cE) =  \frac{1}{3}(1-e^{-6\mu}) + e^{-6\mu}\PP(\cE|J=0).$$
For the last term, we find that:
\begin{equation}
\PP(\cE|J=0) = p_0(\mu)  +(p_1(\mu)+p_3(\mu))  \cdot \frac{(1-e^{-B})}{3} + p_3(\mu) \cdot e^{-B} +p_2(\mu)\cdot 0 
\end{equation}
from which Eqn. ~(\ref{random2}) now follows, by Lemma~\ref{helps}.
The remainder of part (ii) is justified by  Eqn.~(\ref{sym}), and straightforward algebra shows that the expression in Eqn.~(\ref{random2}) is never lower than
the probability of a specific mismatch topology on $A$.

\end{proof}

 \begin{figure*}[h]
 \begin{center}
\resizebox{8cm}{!}{
\includegraphics{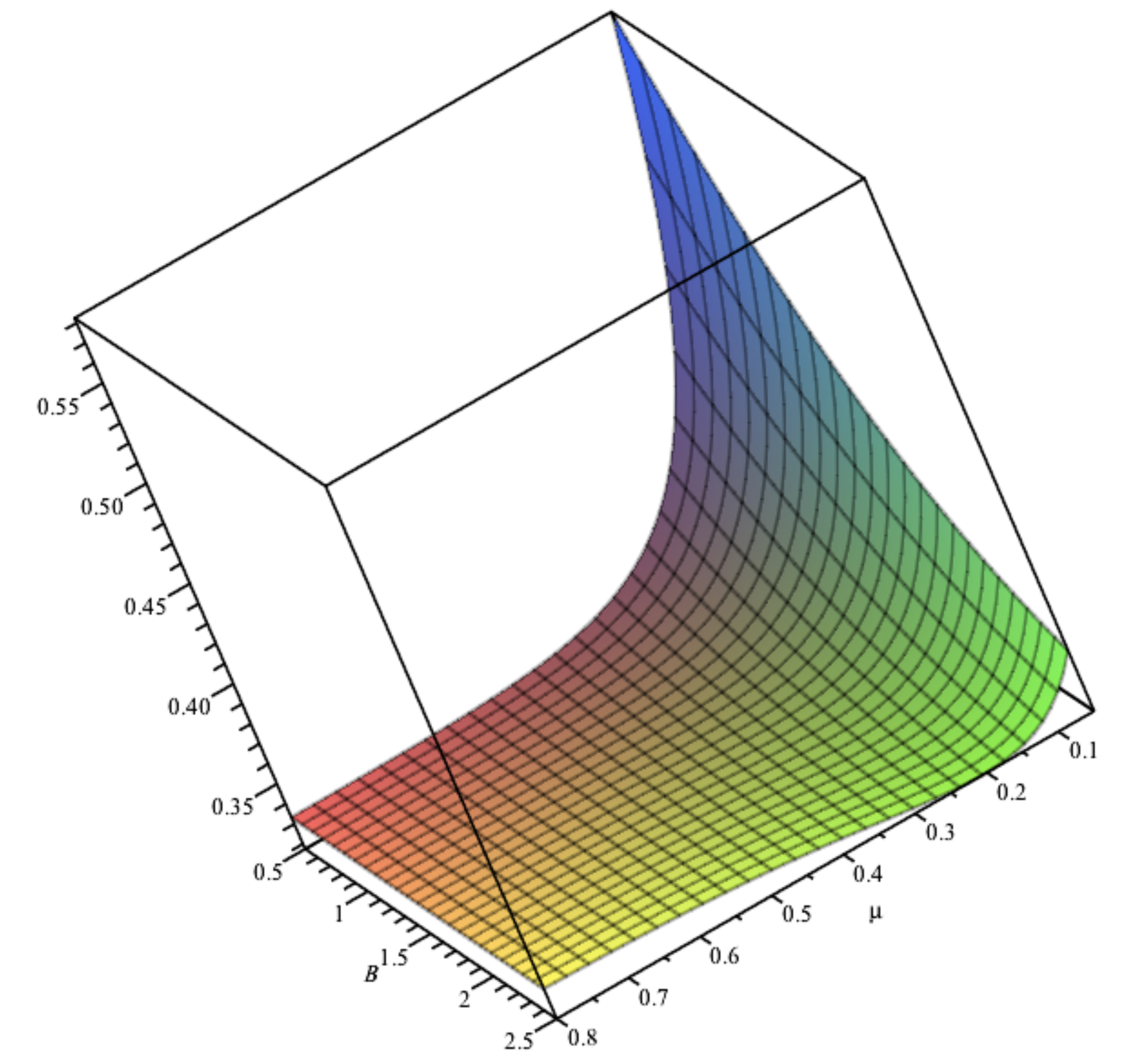}}
\caption{Plot of Equation 1 in the statement of Theorem~\ref{fourcasethm}.  Note that the probability $\mathbb{P}(a|bc)$ for a match (shown on the $z$-axis) is less than $\frac{1}{3}$ in the bottom right-hand side of the figure.}
\end{center}
 \label{eq1}
\end{figure*}

\bigskip

\subsection{ Statistical inconsistency?}

Part (i) of Theorem \ref{fourcasethm} shows that the species tree topology for three taxa inside a larger tree can have the {\em lowest} probability among the three possible gene tree topologies on those three taxa,  under the standard LGT model (see Fig. 5).  This is in sharp contrast to what occurs with incomplete lineage sorting, where the most probable gene tree for three taxa matches the species tree topology for those three taxa, regardless of what other taxa are present, and how they are arranged in the species tree.  Thus, in the setting of Theorem \ref{fourcasethm}(i),  estimating the 
species tree for a set $A$ of three taxa from the frequency of triplet gene trees will be statistically inconsistent (it will converge on an incorrect tree). However, this does not imply that one cannot estimate
the species tree from the probability distribution of {\em all} gene trees topologies (and for all subsets $A$ of size three from $X$).  Moreover, if we use the $R^*$ tree reconstruction method for four-taxon tree
having the clades $\{a,b\}, \{c,d\}$, Theorem \ref{fourcasethm}  (parts (i) and (ii)) shows that we will always return a tree that includes at least one of these clades, and no contradictory clades.
Thus, the $R^*$ method would, in this case, be only weakly inconsistent (i.e. it would return a tree that is either equal to or is resolved by the species tree, rather than being positively misleading). 

\bigskip

\begin{figure*}[ht]
\center
\resizebox{10cm}{!}{
\includegraphics{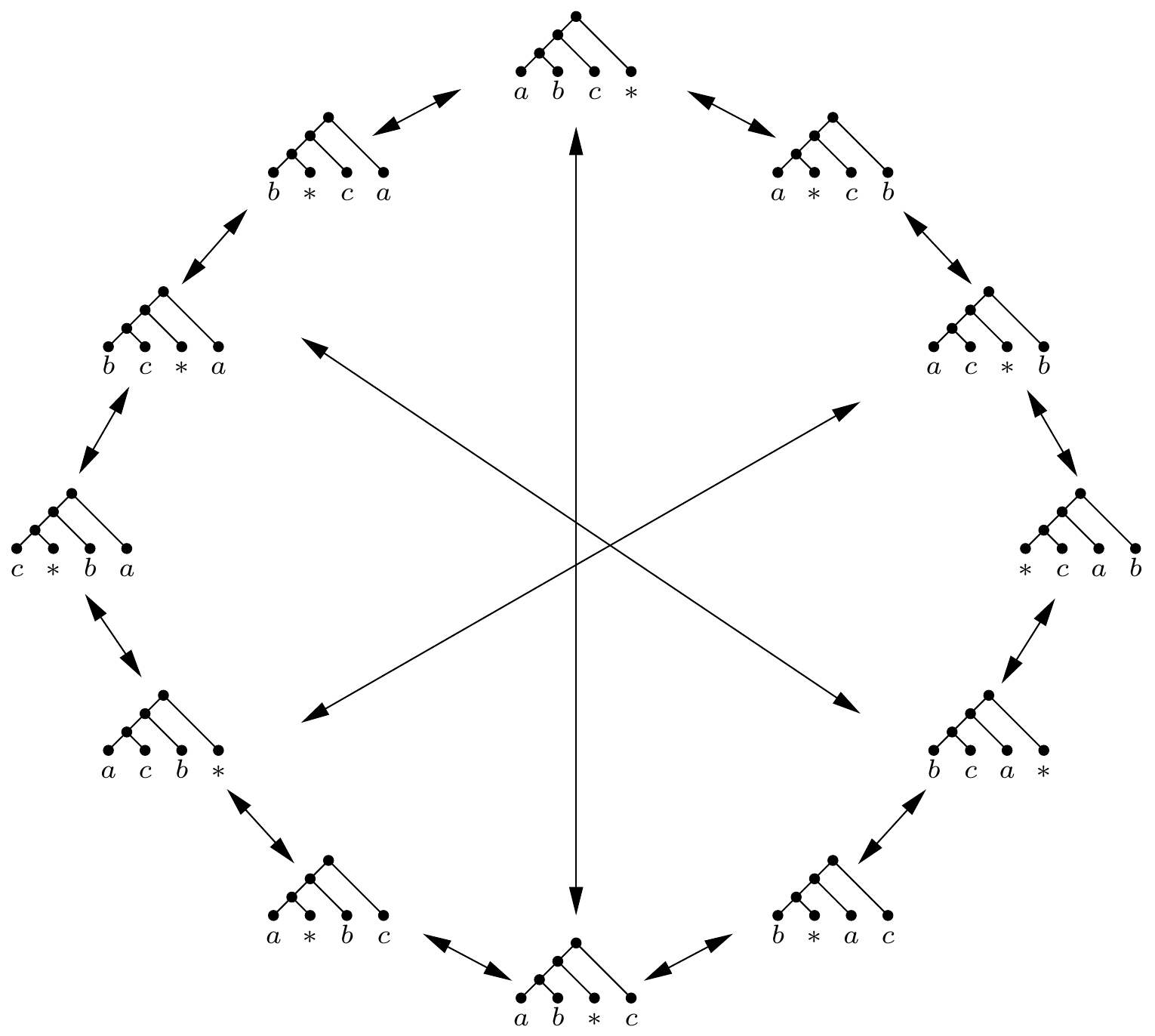}
}
\caption{The graph corresponding to Fig.~\ref{fig:RandomWalk} for the tree that has one cherry.  }
\label{fig:12cycle}
\end{figure*}

  \subsection{The other tree shape on four taxa}

One could perform a similar analysis for the 12 rooted binary trees that have the four leaves $\{a,b,c, *\}$ and just a single cherry (rather than two cherries as above).
  In this case,  the associated transition graph consists of a 12-cycle, together with three additional edges -- obtained by placing an edge between $((xy)z)*$ and $((xy)*)z$ 
  for each of the three choices of $\{x,y\}$ from $\{a,b,c\}$. This graph is shown in Fig.~\ref{fig:12cycle}. The analysis of the probability of a matching topology for $a,b,c$ under the random LGT model (depending on the position of the * lineage) could be  carried out by a similar, albeit more complex, analysis to  that for the simpler 6-cycle graph, but this is beyond the scope of the current paper.

  \section{General case, trees with $n-$taxa}
  \label{general}

We now include the three- and four-taxon results into a more comprehensive statement concerning the statistical consistency of species tree reconstruction under the LGT model, for an
arbitrary number of taxa. 

\begin{theorem}
\label{mainthm}
Consider the standard or extended LGT model on a rooted binary phylogenetic $X-$tree.
\begin{itemize}
\item[(i)] If  $T$ has just three taxa, then under either model the probability that a transfer sequence  induces a match for the three taxa is strictly greater than the probability that it induces  either one of the two mismatch topologies (which have equal probability).
\item[(ii)] A four-taxon tree and branch lengths exist for which the model can assign higher probability to a particular  mismatch topology for some triplet, than for a match, even under the standard LGT model of \cite{lin}.
\item[(iii)] Regardless of the number of taxa in the tree and the branch lengths if, for some subset $A$ of taxa of size 3, the expected total number of transfers into an $A-$lineage  (for the particular gene) is no more than 0.69 in the extended model, and no more than 
1.14 in the standard LGT model, then the probability of a topology match is strictly greater  than the probability of either of the mismatch topologies.
\item[(iv)] When LGT rates ensure that condition (iii) holds for every subset of $A$ of taxa of size 3, there is a polynomial-time method for reconstructing the species tree from the gene trees which is statistically consistent under the model, as the number of independently generated gene trees tends to infinity.
\end{itemize}
\end{theorem}

\bigskip

\begin{proof}

Parts (i) and (ii) are established by Proposition~\ref{lems} and Theorem~\ref{fourcasethm} respectively. 

\bigskip

\noindent {\em Proof of part (iii)}:  Let $N_A$ denote the total number of  transfer events of the gene in the tree into an $A-$lineage.   By Lemma~\ref{lempois}, $N_A$ has a Poisson distribution with some mean $m$. Then for any triple $a,b,c$,  suppose that $T|\{a,b,c\} = a|bc$.  
Then if $N_A=0$ then there is a match with probability $1$.  Thus, in the extended model, if $m<\ln(2) \approx 0.69$ then the probability of a match is at least
$\PP(N_A=0) = e^{-m} >0.5$. This establishes the first claim in part (iii).

For the second claim in part (iii),  consider the standard LGT model.  In the Appendix we establish the following claim by means of a coupling-style argument. 

{\bf Claim 1:}  Consider a sequence of transfer events under the standard LGT model. Then, conditional on the event that $N_A=1$, the probability $p$ that this sequence of transfers induces a match  for $A$  is greater or equal to the probability $q$ of inducing a specific mismatch topology (say $c|ab$) for $A$.

Now, under  the standard LGT model, the probability of a match is at least:
\begin{equation}
\label{pp}
\PP(N_A=0) + p \cdot \PP(N_A=1),
\end{equation}
while the probability of a specific mismatch topology (say $c|ab$) is at most
\begin{equation}
\label{qq}
q \cdot  \PP(N_A=1) + \PP(N_A>1),
\end{equation}
and from Claim 1,  $p \geq q$ so the difference obtained by subtracting the mismatch topology probability (\ref{qq}) from the matching topology probability  (\ref{pp}) is at least:
$$\PP(N_A=0) - \PP(N_A>1) = e^{-m} -(1-e^{-m}-me^{-m}) = e^{-m}(2+m) -1,$$
and the term on the right is strictly positive for $m \leq 1.14$, as claimed.

\bigskip 
\bigskip

{\em Proof of part (iv)}: We can apply the same argument used by  \cite{deg}, who showed that the (polynomial time) 
 $R^*$ tree reconstruction method (based on triplet topologies)  is statistically consistent under models of incomplete lineage sorting.
Here, we are dealing with LGT rather than incomplete lineage sorting, but the only property required of either model in order to ensure
the statistical consistency of the $R^*$ method is that for each triplet $A$, the probability that the gene tree matches the species tree
topology restricted to $A$ has a probability that is greater (by some fixed $\epsilon>0$) than either of the other two topologies (in the case of incomplete lineage sorting,
the two alternative topologies have equal probability, but this may not be the case under LGT; however, this is not essential to prove  consistency).
\end{proof}

\subsection{Rates of LGT}

Theorem~\ref{mainthm} (part (iii)) requires a small expected number of  transfers into an $A-$lineage for any subset $A$.  The question arises as to how this expected number would compare with the 
total expected number of transfers in the tree.  The total number of LGT transfers in the tree $N_{\rm tot}$ has a Poisson distribution; under the standard LGT model   this distribution 
has mean $r \cdot L(T)$, where $r$ is the rate of LGT transfer out of any given lineage at any given time and  $L(T)$ is the phylogenetic diversity of $T$ (the sum of the lengths of all its branches).   For a Yule (pure-birth) tree with $n$ taxa and timespan $t$, if the speciation rate is set equal to its expected value, then from \cite{ste}, $L(T)$ has expected value:
\begin{equation}
\label{expPD}
L_n(t) = \frac{(n-2)t}{\ln(n/2)}.
\end{equation}

On the other hand,  the rate of transfers into an $A-$lineage at any time is at most $3r$ and so the expected number of transfers into an $A-$lineage is, at most, $3rt$.  Thus, if the expected number of
LGT transfers in the entire tree is $G$ then the expected number of transfers into an $A-$lineage in a tree with phylogenetic diversity equal to its expected value under the Yule model  is, at most:

\begin{equation}
\label{expPD2}
\frac{3\ln(n/2)}{(n-2)} \cdot G.
\end{equation}

For example, if $n=200$ and $G\leq 10$ (on average every gene is transferred at most 10 times on the tree) then (\ref{expPD2}) takes the value 0.7, which is within the
1.14 bound of Theorem~\ref{mainthm} (part (iii)). We note that in one study, it was suggested  that the average number of times each gene has been transferred might be around 1.1 \cite{dag}, so the condition imposed in Theorem~\ref{fourcasethm} (iii) may  not be unreasonable. In the recent paper by \cite{abb}, an average rate across bacterial genomes was estimated at between .02 to .04 LGTs per branch of the tree. Thus, for $n=200$, $G$ is about 8-16 events, where the upper end of these results include some cases of extremely high rates of LGT in bacteria.

\subsection{LGT and incomplete lineage sorting}

We have described sufficient conditions for the $R^*$ tree-reconstruction method to be a statistically consistent estimator of a species tree topology under various LGT models. Moreover, it is known that the $R^*$ method is also a statistically consistent estimator of the species tree topology under lineage sorting and without any non-trivial restrictions on branch lengths \cite{deg}.    It follows that if each topology of each gene tree is determined by either  LGT acting on the species tree topology or by  incomplete lineage sorting (but not by both processes) then the  $R^*$ tree reconstruction method can be a statistically consistent estimator of species tree topology (under conditions where it would be for the genes undergoing LGT).   

One could also ask what happens when the two processes are combined -- that is,  if we allow the ancestry of a gene to follow transfers and
to coalesce within lineages, under the usual coalescent process.    

While it is possible to extend the earlier results a little in this direction (results not shown), the applicability of the results is somewhat limited for the following reason.

LGT is especially prevalent in haploid, largely asexual taxa with limited recombination, such as bacteria and archaea and, in this case, although incomplete lineage sorting may apply in considering how genetic lineages coalesce in the species tree, there is an important difference to diploid sexual taxa. Namely, in this latter case, the coalescent history of each gene sampled from an extant individual in each species represents an essentially independent sample from the same (multi-species coalescent) process, while in taxa with little or no recombination, the lineages of all genes follow the same ancestral trajectory, apart from LGT events.  This complicates any statistical analysis based on assuming that the genes are independent  samples from a common process  and leads to some delicate statistical issues in attempting to analyse data with such a mixed mode.  We defer this issue for future consideration.

\section{Missing taxa: Primordial tree consensus}

One obstacle to applying the $R^*$ construction is that many genes may not be present across all taxa \cite{san}. This may be due to a variety of factors, including  gene loss or gene conversion, or simply because certain genes have not been sequenced yet. 

Consequently, we describe a slight extension of the $R^*$ consensus approach to handle this situation. 
For any three taxa $a,b,c \in X$, let $G(a,b,c)$ denote the set of genes that are present in all three taxa $a,b,c$.  We will assume that the pattern of taxon coverage is sufficiently 
dense that the following condition holds:
\begin{equation}
\label{Geq}
G(a,b,c)>0, \mbox { for all } a,b,c \in X.
\end{equation} 

To give some indication of how much coverage this requires,  if $n_i$ denotes the number of taxa containing gene $i \in \{1, \ldots, k\}$, and $n$ is the total number of taxa, then we require:
$$\sum_{i=1}^k n_i(n_i-1)(n_i-2) \geq n(n-1)(n-2),$$
(from the proof of Theorem 1 of \cite{san}),
which, in turn, implies the weaker but simpler, inequality:
 $$\sum_{i=1}^k f_i^3  \geq 1 \mbox{ for } f_i = n_i/n.$$

We consider the following simple extension of the $R^*$ consensus method to the setting of partial taxon coverage.  For  $a,b,c \in X$, let $cf(a|bc)$ denote the proportion of genes in $G(a,b,c)$ which resolve $a,b,c$ as the triplet topology $a|bc$.  

\begin{lemma}
The set $$C= \{A \subseteq X: cf(b|aa') > \max\{cf(a|a'b), cf(a'|ab)\} \mbox{ for all } a,a' \in A, a\neq a' \mbox{ and }  b \in X-A\}$$
forms a hierarchy, and can be constructed in time that is polynomial in $n$.
\end{lemma}

This procedure has been implemented in the phylogenetic software package {\em  Dendroscope 3} (version 3.2.2) \cite{den} as the `primordial tree' consensus method. 

Now, under a model in which the pattern of gene presence and absence is a random process (as in \cite{san}, where the presence or absence of a gene for each taxon is an independent stochastic process) and provided this process  is independent of the LGT process,
the results on the statistical consistency of species tree reconstruction will carry over.  The same also holds if we were to consider incomplete lineage sorting rather than LGT.

Fig.~\ref{rice} shows a tree constructed in this way from the recent bacterial data set of \cite{abb}, for which the authors estimated that the rate of LGT was fairly high (but  not too high to erase all phylogenetic signal).  The input was 1338 rooted gene trees on variable label sets for the Actinobacteria phylum of their study (the clade at the upper left in Fig. 3 of \cite{abb}).

 \begin{figure}[h]
  \begin{center}
\resizebox{12cm}{!}{
\includegraphics{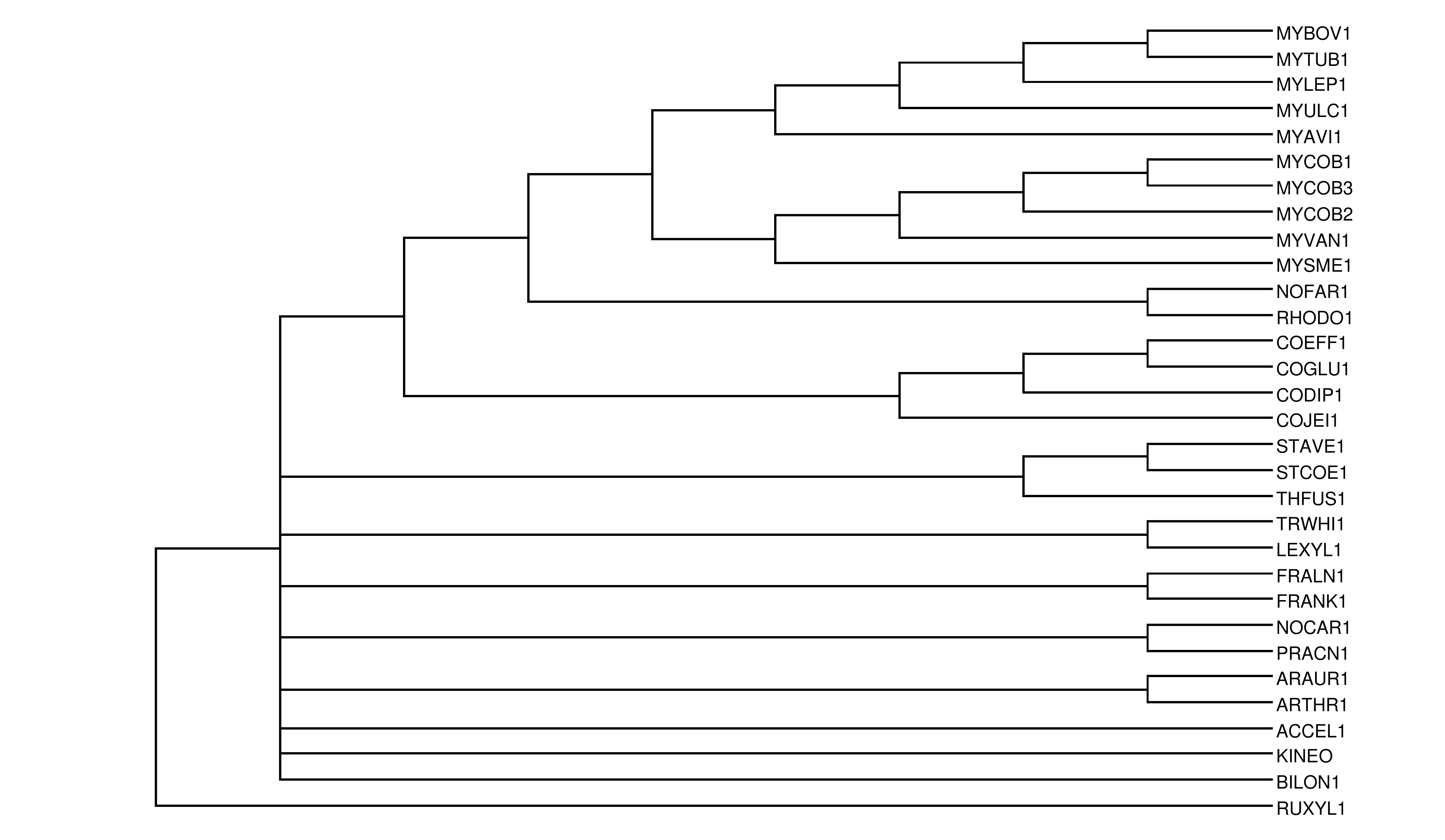}}
\caption{A tree constructed from 1338 rooted gene trees on overlapping taxon sets for the Actinobacteria phylum, which indicated high rates of LGT in the study by \cite{abb}.}
\end{center}
 \label{bacteria}
\end{figure}

The `unrooted'  gene trees (from the website of the authors) were midpoint-rooted using the phylogenetic program `Phylip'. The authors in \cite{abb} implemented a complex procedure as part of their `Prunier' software, to look at all possible rootings of the input gene trees, selecting the ones that minimized the number of LGT events. Rooting them in a new way here provides an independent analysis of these data.

The output tree is fairly close to the tree suggested by \cite{abb}.  The biggest difference is that the primordial tree roots it at  {\em R.  xylanophilus}, which is a very long branch in their Fig. 3.  The authors of \cite{abb} were concerned about long branch attraction in these data.

As a second application to a quite different data set (Eukaryotes) and a different process (incomplete lineage sorting rather than LGT),  
Fig.~\ref{rice} shows a tree constructed by the same method,  from 986 rooted gene trees from chromosome 3 of 11 taxa of the genus {\em Oryza} (rice and its relatives). 
Some trees contained all 11 taxa but most did not, however the pattern of taxon coverage is sufficiently dense that condition (\ref{Geq}) does hold. 
While this data set is unlikely to exhibit LGT at anywhere near the rate of the previous one, gene flow persisting for some time after speciation would essentially show the same pattern as LGT. Moreover, incomplete lineage sorting is likely quite extensive across the genome at several nodes in the tree (\cite{zou}, \cite{zwi}); Zwickl et al., in prep.).

 \begin{figure}[h]
  \begin{center}
\resizebox{10cm}{!}{
\includegraphics{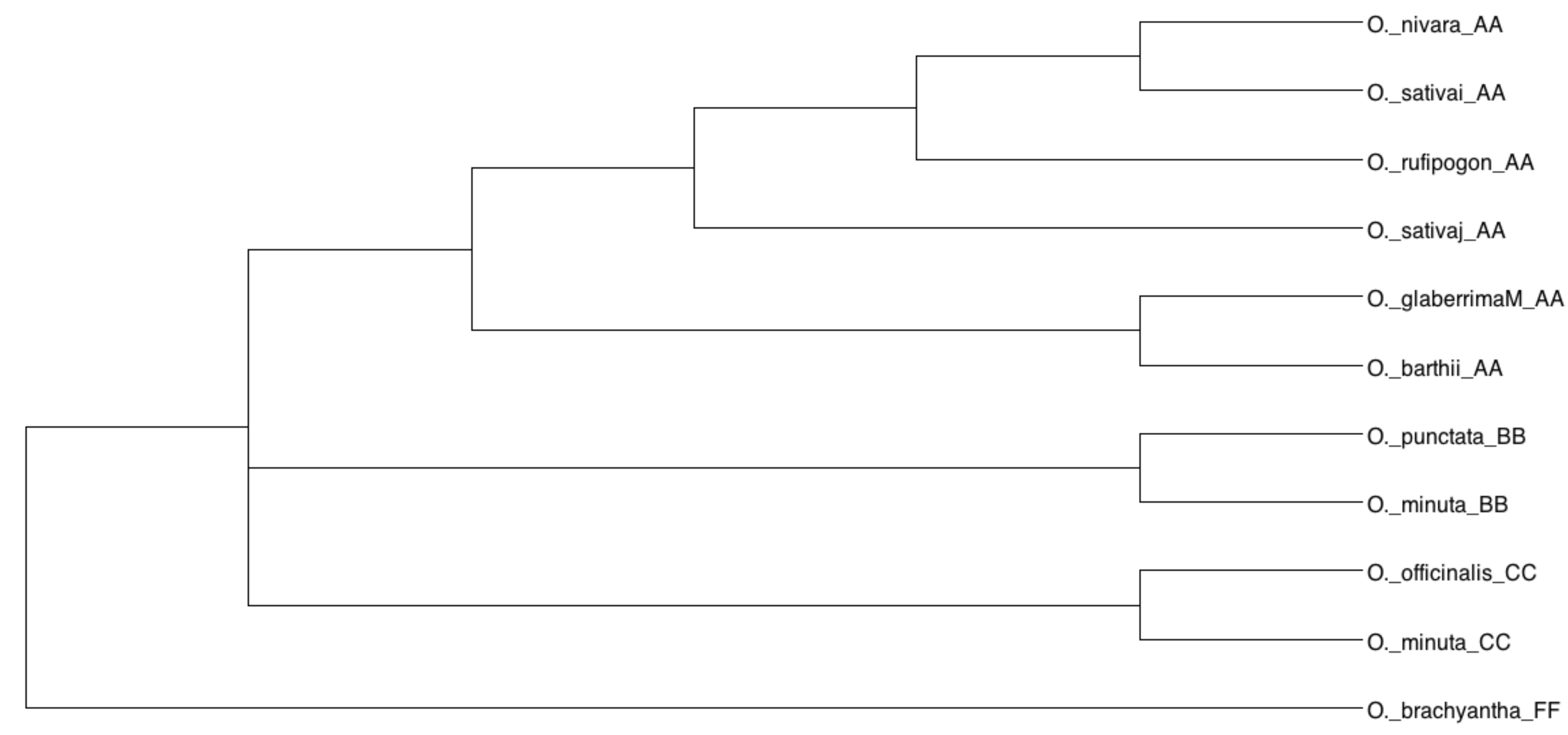}}
\caption{A tree constructed from 986 rooted gene trees from chromosome 3 for taxa of the genus Oryza for which incomplete lineage sorting (rather than LGT) is a likely cause of gene tree discordance.}
\end{center}
 \label{rice}
\end{figure}

\subsection{Questions for future work}

It would be interesting to extend  the scope of Theorem~\ref{fourcasethm} to include  the other four-taxon tree-shape (under the standard LGT model), as well as to analyze both tree shapes under the extended LGT model.   Our analysis also raises some intriguing statistical questions:  Is strong statistical inconsistency possible (for R$^*$ or perhaps other methods)?  Is the species tree identifiable from the probability distribution on gene trees, regardless of the LGT rate?  If the rate of LGT decreases sufficiently fast with phylogenetic distance in the tree, then is statistical consistency restored for the R$^*$ method? And what can be said regarding statistical issues arising when we combine LGT and incomplete lineage sorting and phylogenetic sampling error?  Further research on these questions would help us better understand the extent to which signal for a species tree can be recovered above the `noise' of random processes that can cause gene trees to conflict with the species tree.

\subsection{Acknowledgments}  S.L. was supported by a Marie Curie International Outgoing Fellowship within the 7th European Community Framework Programme.  M.S. was supported by the  Allan Wilson Centre for Molecular Ecology and Evolution.   We also thank Derrick Zwickl and the Oryza Genome Evolution Project.

\section{Appendix: Proof of Claim 1 (from proof of Theorem \ref{mainthm})}

Consider a sequence $\underline{\sigma}$ of one or more transfer events, which contains exactly one transfer $\sigma_r=\sigma(p, p')$ that is into an $A-$lineage and for which
$T[\underline{\sigma}]|A = c|ab$. 
We will associate with $\underline{\sigma}$ another sequence of transfer events $\underline{\sigma}'$ which induces a match for $A$, and which is identical to
$\underline{\sigma}$ except that $\sigma$ is substituted by a particular alternative transfer $\sigma'$.

In case $\sigma$ is an $A-$joining transfer (in which case it joins $a$ to $b$, or $b$ to $a$ in order for $T[\underline{\sigma}]|A = c|ab$) we replace $\sigma$ with the
transfer $\sigma'$ (at the same time-instant)  that:
\begin{itemize}
\item[(i)]  joins $b$ to $c$ if $\sigma$ joins $a$ to $b$;
\item[(ii)]  joins $c$ to $b$ if $\sigma$ joins $b$ to $a$. 
\end{itemize}
These two cases are illustrated in Fig.~\ref{fig_five_case}.

\begin{figure*}[ht]
\center
\resizebox{14cm}{!}{
\includegraphics{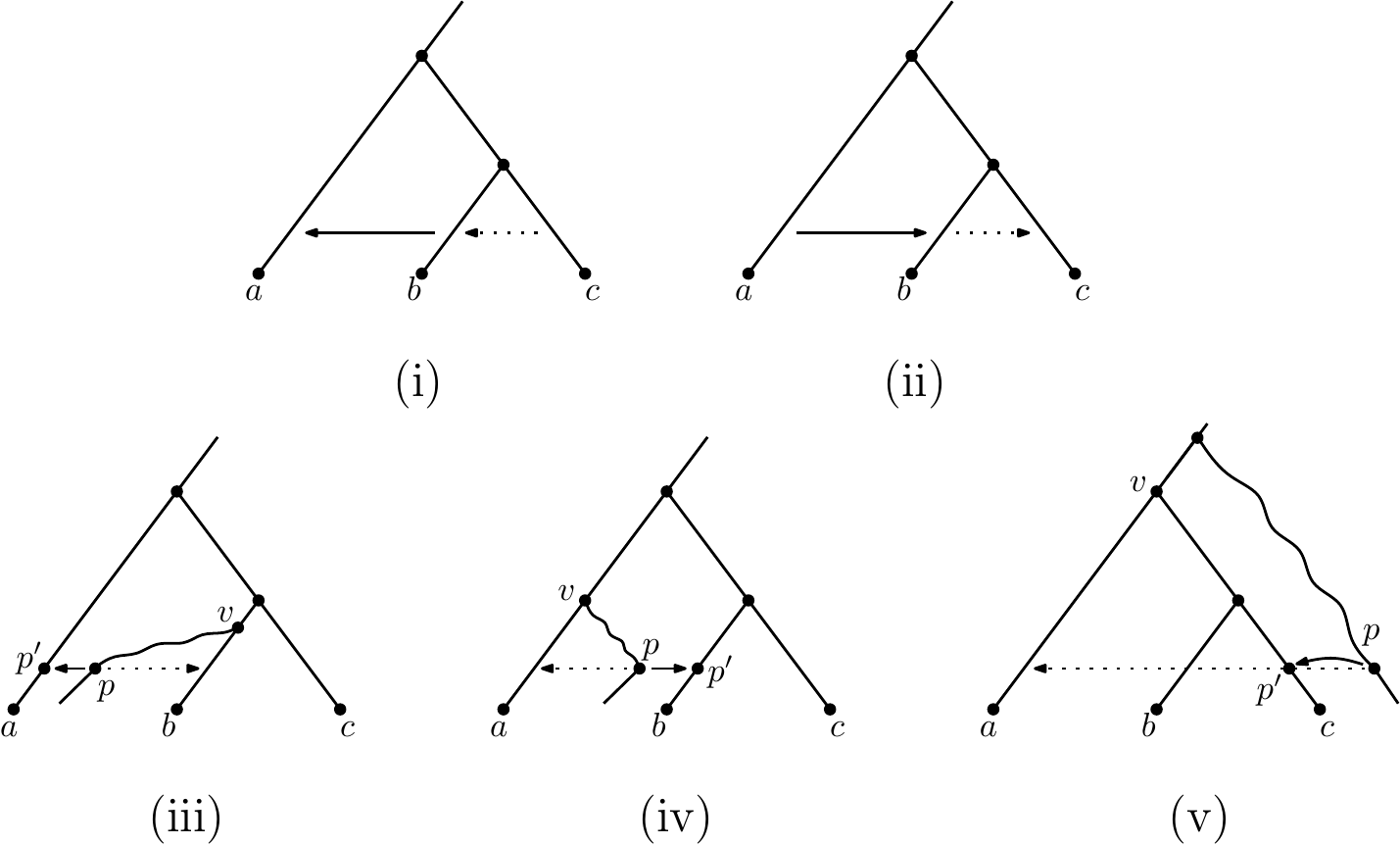}
}
\caption{For cases (i)--(v) in the proof of Claim 1,  the transfer $\sigma$ is shown as a solid horizontal arrow, and the associated transfer $\sigma'$ is shown as a dashed horizontal arrow.}
\label{fig_five_case}
\end{figure*}

Otherwise, in case $\sigma=(p,p')$ is an $A-$moving transfer, either: 
 \begin{itemize}
 \item [(iii)] ${\rm des}_A(T, p')=\{a\}$, or 
 \item [(iv)] ${\rm des}_A(T, p')=\{b\}$, or
 \item [(v)] ${\rm des}_A(T, p')=\{c\}$.
 \end{itemize}

 Consider case (iii). Recalling that $\sigma$ is the $r$-th transfer in $\underline{\sigma}$, 
in order for $\underline{\sigma}$ to induce the topology $c|ab$, there must be a vertex $v$ in the derived trees $T_{r-1}$ for which ${\rm des}_A(T_{r-1}, v)=\{a, b\}$; moreover, since $\sigma$ is the only transfer into an $A-$lineage, $v$ must lie on the the path in $T$ between $b$ and the MRCA of $b,c$.
In this case, we take $\sigma'= (p, p_b)$ where $p_b$ is the unique point in $T$ with $t(p_b)=t(p)$ and which has ${\rm des}_A(T, p_b)=\{b\}$.

In case (iv), in order for $\underline{\sigma}$ to  induce the topology $c|ab$ there must be a vertex $v$ in the derived trees $T_{r-1}$ for which  ${\rm des}_A(T_{r-1}, v)=\{a, b\}$; moreover, since $\sigma$ is the only transfer into an $A-$lineage, $v$ must lie on the path in $T$ between the $a$  and the MRCA of $A$. 
In this case, we take $\sigma'= (p, p_a)$ where $p_a$ is the unique point in $T$ with $t(p_a)=t(p)$ and which has ${\rm des}_A(T, p_a)=\{a\}$.

Case  (v) is similar to case (ii) except that we can take $v$ to be the MRCA of $\{a,b\}$, and (as in case (ii)) we  take $\sigma'= (p, p_a)$ where $p_a$ is the unique point in $T$ with $t(p_a)=t(p)$ and which has ${\rm des}_A(T, p_a)=\{a\}$.

Cases (iii)--(v) are also shown in Fig.~\ref{fig_five_case}. 

In all five cases (i)--(v),  replacing $\sigma$ by $\sigma'$ in the sequence $\underline{\sigma}$ results in the modified sequence  $\underline{\sigma}'$ that induces a match for $A$;
moreover the association  $\underline{\sigma} \mapsto \underline{\sigma}'$ is one-to-one on the  set of transfer sequences with $N_A=1$ and which induce
$c|ab$ and $a|bc$ respectively. It then follows that  $p \geq q$ under the standard LGT model, as claimed.

\end{document}